\RequirePackage{xcolor}
\documentclass[letterpaper, 10 pt, conference]{IEEEconf}
\IEEEoverridecommandlockouts
\usepackage{circuitikz}
\usepackage{cite}

\usepackage{mathtools,amsthm,amsmath,mathalfa,bbm,amssymb}
\usepackage{comment}
\usepackage{textcomp}
\makeatletter
\def\endfigure{\end@float}
\makeatother
\renewenvironment{proof}{
\noindent{\sc Proof.}\hspace{0.10cm}\,\,}{$\hfill\Box$\vspace{.1cm}}

\newcommand{\dd}{\textrm{d}}

\usepackage{algorithmic}
\usepackage{graphicx}
\usepackage{textcomp}
\usepackage{xspace}

\newcommand{\rline}{{\mathbb R}}

\newcommand{\sign}{\, \mathrm{sign}\,}

\newcommand{\rfb}[1]{\mbox{\rm
		(\eref{#1})}\ifx\undefined\stillediting\else:\fbox{$#1$}\fi}

\newcommand{\bluff}{{\hbox{\raise 15pt \hbox{\hskip 0.5pt}}}}

\newfont{\roma}{cmr10 scaled 1200}

\newtheorem{thm}{Theorem}[section]

\newtheorem{prop}[thm]{Proposition}

\newtheorem{defn}{Definition}[section]

\def\BibTeX{{\rm B\kern-.05em{\sc i\kern-.025em b}\kern-.08em
    T\kern-.1667em\lower.7ex\hbox{E}\kern-.125emX}}
\begin{document}

\title{On Port-Hamiltonian Formulation of Hysteretic Energy Storage Elements: The Backlash Case}
\author{J.R. Keulen, B. Jayawardhana, A.J. van der Schaft 
 \thanks{This research project is supported by the TKI (Topconsortia voor Kennis en Innovatie) grant 22.0027 under the Top Sector High-Tech Systems and Materials (HTSM).}
 \thanks{J.R. Keulen and B. Jayawardhana are with the Engineering and Technology Institute Groningen, Faculty of Science and Engineering, University of Groningen, 9747AG Groningen, The Netherlands (e-mails: j.r.keulen@rug.nl, b.jayawardhana@rug.nl)}
 \thanks{A.J. van der Schaft is with the Bernoulli Institute for Mathematics, Computer Science, and
    Artificial Intelligence, University of Groningen, Nijenborgh 9, 9747 AG
    Groningen, Netherlands (e-mail: a.j.van.der.schaft@rug.nl).}}

\maketitle
\thispagestyle{empty}
\begin{abstract}
This paper presents a port-Hamiltonian formulation of backlash-driven hysteretic energy storage elements. 
First, we revisit the passivity property of backlash-driven storage elements by presenting a family of storage functions. We explicitly derive the corresponding available storage and required supply functions \`a la Willems \cite{willems1972dissipative}, and show the interlacing property of the aforementioned family of storage functions sandwiched between the available storage and required supply functions. 
Second, using the obtained family of storage functions, we present a port-Hamiltonian formulation of hysteretic inductors. In particular, we show how an appropriate Hamiltonian function is defined using the family of storage functions and how the hysteretic elements can be expressed as port-Hamiltonian system with feedthrough term, where the feedthrough term represents energy dissipation. Correspondingly, we illustrate its applicability in describing an RC circuit (in parallel and in series) containing a backlash inductor. 
\end{abstract}


\section{Introduction}\label{sec:introduction}

Modern high-precision technology increasingly operates at the nano- and microscale. At these scales, hysteresis is a fundamental limiting factor for system performance, stability and efficiency. Hysteresis is a phenomenon where the current state does not only depend on the current input, but also on the history of previous inputs. Such behavior is encountered in a wide spectrum of physical systems that contain ferromagnetic-, piezoelectric materials, magnetostrictive actuators, shape-memory alloys, etc. \cite{brokate2012hysteresis, mayergoyz2003mathematical}. In thermodynamics and in electro-mechanical systems, the presence of hysteresis loop is always associated to energy dissipation, where the enclosed area of a hysteresis loop 
represents the dissipated energy per cycle. Over the past decades, accurate modeling of hysteresis phenomena has been investigated for analyzing systems performance, for compensating the hysteresis through the deployment of inverse hysteresis and for embracing hysteresis as a new class of set-and-forget actuators \cite{huisman2021,beltran2021,keulen2024}. 

Several mathematical frameworks have been proposed to describe hysteresis, including operator-based models such as the Preisach and Prandtl-Ishlinskii (PI) models \cite{brokate2012hysteresis}, and differential equation-based models such as the Duhem models \cite{ikhouane2018survey}. The backlash operator, also the fundamental building block of the PI model, belongs to a subfamily of Duhem operators with counterclockwise input-output dynamics \cite{padthe2005counterclockwise, jayawardhana2012stability}, a property that is  connected to passivity in \cite{gorbet1998generalized}, where it is shown that hysteretic systems do not generate energy but dissipate it over each cycle. In the dissipativity framework of Willems \cite{willems1972dissipative}, a system is dissipative if there exists a storage function satisfying a dissipation inequality with a given supply-rate function. For hysteretic systems, the characterization of such storage functions is non-trivial due to the path-dependent, multivalued nature of the input-output map, where a single input value may correspond to multiple outputs. While passivity of the backlash operator has been established \cite{jayawardhana2009sufficient}, an explicit characterization of the family of admissible storage functions, bounded by the available storage and required supply functions, has not been addressed in the literature. Moreover, the non-conservative nature of hysteresis makes the unified formulation into an energy-based framework such as port-Hamiltonian systems particularly challenging.

Port-Hamiltonian (pH) systems provide a structured, energy-based, modeling framework that has proven effective for physical systems across multiple domains \cite{van2014port, duindam2009modeling}. However, physical elements exhibiting hysteresis behavior cannot be represented by a single storage or dissipation component alone \cite{duindam2009modeling}. In \cite{karnopp2012system, karnopp1983computer}, hysteretic systems are modeled by introducing an additional nonlinear dissipative element to capture hysteretic energy dissipation. In recent years, irreversible pH formulations have been proposed \cite{ramirez2013irreversible}, where the dissipated energy is accounted for as thermal energy (entropy). A modular passivity based pH of piezo actuators has been proposed in \cite{diaz2025modular}, using 'hysterons' as in \cite{karnopp1983computer} to capture the hysteretic behavior.  Nevertheless, a systematic pH formulation that directly incorporates hysteretic dissipation through a nonlinear potential, consistent with the Willems dissipativity framework, remains an open problem.

In this paper, we address these challenges by developing a pH formulation of backlash-driven storage elements. By considering the hysteretic phenomenon in a backlash inductor, we present a family of storage functions that is bounded from below by the available storage and from above by the required supply in the sense of Willems \cite{willems1972dissipative}, both of which we compute explicitly. Based on these results, we show how to define an appropriate Hamiltonian function that allows us to cast the hysteresis in the backlash inductor into a general port-Hamiltonian formalism that includes feedthrough terms. Finally, using the compositionality of pH systems, we include the backlash inductor as an element of an interconnected system, which we illustrate for the parallel and series interconnection with a
linear resistor and capacitor. 

\section{Preliminaries}
\subsection{Dissipativity Theory}\label{sec:dissipativity}
Consider a dynamical system $\Sigma$ \cite[Definition 1]{willems1972dissipative} with state space
$\mathcal{X} \subseteq \rline^n$, input space $\mathcal{U} \subseteq \rline^m$, and output space $\mathcal{Y} \subseteq \rline^m$, described by its set of admissible trajectories $(u(\cdot), x(\cdot), y(\cdot))$ on intervals of $\rline$.
A supply rate $w : \mathcal{U} \times \mathcal{Y} \to \mathbb{R}$ quantifies the rate at which generalized power is exchanged via the external ports of $\Sigma$. 

\begin{defn}{\rm \cite[Definition 2]{willems1972dissipative}} 
\label{def:dissipativity}
The system $\Sigma$ is \emph{dissipative} with respect to the supply rate
$w$ if there exists a non-negative function
$S : \mathcal{X} \to \mathbb{R}_{\geq 0}$, called a
\emph{storage function}, such that the \emph{dissipation inequality}
\begin{equation}
    S\bigl(x(t_1)\bigr)
    \;\leq\;
    S\bigl(x(t_0)\bigr)
    + \int_{t_0}^{t_1} w\bigl(u(t),\, y(t)\bigr)\, \mathrm{d}t
    \label{eq:dissipation_ineq}
\end{equation}
holds for all $t_1 \geq t_0$ and all admissible trajectories.
If $S$ is differentiable, \eqref{eq:dissipation_ineq} is equivalent to
$\dot{S}(x) \leq w(u, y)$.
\end{defn}
The system $\Sigma$ is called \emph{passive} if it is dissipative w.r.t. the supply rate $w(u,y)=u^\top y$. 
The storage function satisfying Definition~\ref{def:dissipativity} is in
general not unique, i.e., there may exist 
other storage functions that satisfy \eqref{eq:dissipation_ineq}. Willems in \cite{willems1972dissipative} characterized this family through two extremal storage functions.
The \emph{available storage}
\begin{equation}
    S^a(x_0)
    = \sup_{\substack{u(\cdot),\; T \geq 0 }}
    -\int\limits_0^{T} w\bigl(u(t), y(t)\bigr)\, \mathrm{d}t, \quad x(0)=x_0
    \label{eq:available_storage}
\end{equation}
is the maximum energy that can be extracted from $\Sigma$ starting at $x_0$, and
the \emph{required supply} 
\begin{equation}
    S^r(x_0)
    =\inf_{\substack{u(\cdot),\ T \geq 0  }}
    \int\limits_{-T}^{0} w\bigl(u(t), y(t)\bigr)\dd t , 
    \label{eq:required_supply}
\end{equation}
where  $x(-T)=x^*,\ x(0)=x_0$, is the minimum energy needed to reach $x_0$ from a ground state
$x^* \in \mathcal{X}$.
These functions bound any admissible storage function $S$ by \cite{van2000l2}
\begin{equation}
    S^a(x) \;\leq\; S(x) \;\leq\; S^r(x) + S^a(x^*)
    \quad \forall\, x \in \mathcal{X}.
    \label{eq:sandwich}
\end{equation}
Furthermore, $\Sigma$ is dissipative if and only if $S^a(x) < \infty$ for all
$x \in \mathcal{X}$~\cite{willems1972dissipative}.

\subsection{Port-Hamiltonian systems}

For describing the hysteretic storage elements as pH systems, we require a general form of pH systems as follows \cite{van2024reciprocity, van2026hopfield}. 
A finite dimensional pH system, with \textit{nonlinear} dissipation, takes the form
\begin{equation}\label{eq:pHsystem}
\begin{bmatrix}
    -\dot{x}\\
    y
\end{bmatrix}=\begin{bmatrix}
    -J(x)& -G(x)\\
    G^{\top}(x)& M(x)
\end{bmatrix}\begin{bmatrix}
    e\\
    u
\end{bmatrix}+
\begin{bmatrix}
    \frac{\partial P}{\partial e}(x,e,u)\\
    \frac{\partial P}{\partial u}(x,e,u)
\end{bmatrix}, 
\end{equation}
where $x \in \rline^n$ is the state, $e=\frac{\partial H}{\partial x}$ is the vector of co-energy variables and $u,y \in \rline^m$ are conjugate port variables. 
The Hamiltonian $H:\rline^n \to \rline$ represents the stored energy and is continuously differentiable. 
The interconnection matrices satisfy $J(x) = -J(x)^\top$, $M(x)= -M(x)^\top$, while $G(x)$ denotes the input matrix, and the dissipation potential $P:\rline^n\times\rline^n\times\rline^m\to\rline$ is such that
\begin{equation}\label{eq:PHdissipation} 
e^\top\frac{\partial P}{\partial e}(x,e,u)+u^\top\frac{\partial P}{\partial u}(x,e,u)\geq0, 
\end{equation}
for all $x, e$ and $u$. Using skew-symmetry of the first matrix block, one can verify
\[
\frac{d}{dt}H - u^\top y = - \begin{bmatrix} e^\top & u^\top \end{bmatrix} \begin{bmatrix} - \dot{x} \\ y \end{bmatrix}=
- e^\top \frac{\partial P}{\partial e} - u^\top \frac{\partial P}{\partial u} \leq 0,
\]
where 
the last inequality follows from the dissipation condition in \eqref{eq:PHdissipation}. This shows that the system is \textit{passive} with supply rate $u^\top y$ with $H$ as storage function, if $H\geq0$.


\section{Passivity of Backlash-driven Storage Elements}
In this section, we introduce a family of storage functions for hysteretic storage elements, which will be described shortly. Furthermore we present explicitly the associated available storage and required supply functions of such hysteretic storage elements. 
We begin by introducing the multivalued mapping $\sign$, defined by
\begin{equation}\label{eq:sign_V}
\sign(V)=
\begin{cases}
\{-1\}, & V<0,\\
\{1\}, & V>0,\\
[-1,1], & V=0.
\end{cases}
\end{equation}
Consequently, algebraic relations involving $\sign(V)$ below are to be interpreted as inclusions, which reduce to equalities whenever $V \neq 0$.
\begin{figure}[b]
    \centering
    \includegraphics[scale=0.8]{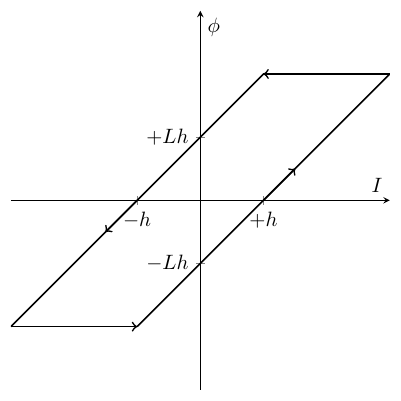}
    \caption{Phase plot of a 'ferromagnetic' backlash inductor, with the variables $I$ and $\phi$. The adjective 'ferromagnetic' here stems from the physical law that we attach to this backlash, where $\phi$ is coupled to a port variable $V$ by $\frac{\dd \phi}{\dd t}=V$. Both branches have slope $L$, which defines the inductance, and they are separated by the play width $2h$ along the $I$-axis. }
    \label{fig:backlash}
\end{figure}

For concreteness, we restrict attention to inductors as the energy storage elements.  Consider a backlash inductor element evolving in the $(I,\phi)$ plane, where $I$ is the \emph{electrical current} and $\phi$ is the \emph{magnetic flux}, as illustrated in Figure \ref{fig:backlash}. In this case, the electrical current and the magnetic flux are described by a backlash operator with width $2h$ and slope $L$, which is associated with the inductance. We note that the backlash behavior shown in Figure \ref{fig:backlash} can be considered an idealization of the hysteresis behavior commonly found in inductors with ferromagnetic cores. The corresponding physical equation for nonlinear inductors is the Faraday induction law, which couples the voltage $V$ to the magnetic flux $\phi$ by
\begin{equation}\label{eq:induction_law}
\frac{\dd \phi}{\dd t}=V. 
\end{equation}
Throughout this paper, $I$, $\phi$ and $V$ are functions of time. For brevity we omit the
argument $t$ unless it is needed for clarity. The pair of variables $(I,V)$ are the port variables of the inductor. 
Let us firstly present a family of storage functions for the backlash inductor as follows.
\begin{prop}\label{prop:1}
For any $\gamma\in\left[-Lh, Lh\right]$, the function
\begin{equation}\label{eq:S_gamma}
    S_\gamma(\phi)=\begin{cases}
        \frac{1}{2L}\left[(\phi+\gamma)^2-\left(\frac{Lh+\gamma}{2}\right)^2\right] & 
        \forall\phi>\frac{Lh-\gamma}{2},\\
         \frac{1}{2L}\left[(\phi-\gamma)^2-\left(\frac{Lh+\gamma}{2}\right)^2\right] & 
         \forall\phi<-\frac{Lh-\gamma}{2},\\
         0 & \text{elsewhere,} 
    \end{cases}
\end{equation}
defines an admissible storage function for the backlash inductor with minimum value of $0$.
\end{prop}

\begin{proof}
We will prove Proposition \ref{prop:1} by evaluating the time-derivative of $S_\gamma$ along admissible trajectories of the backlash inductor. 
Firstly, let us fix $\gamma \in [-Lh,Lh]$. In order to evaluate the time-derivative of $S_\gamma$, we consider two different cases following the definition of $S_\gamma$ in \eqref{eq:S_gamma} as follows. 

\medskip
\noindent\textit{Case 1: $\phi>\frac{Lh-\gamma}{2}$},  for this case we have 
\begin{equation}\label{eq:sdot1}
\frac{d}{dt}S_\gamma(\phi)= \frac{1}{L}(\phi + \gamma)\dot{\phi} = \frac{1}{L}(\phi + \gamma)V
\end{equation}
When $V=0$ it follows immediately that $\frac{d}{dt}S_\gamma(\phi)=0$. 
On the one hand, when $V<0$, which corresponds to the downward motion along the left line of the backlash, the RHS of \eqref{eq:sdot1} satisfies 
\[
\frac{1}{L}(\phi + \gamma)V = \frac{1}{L}(\phi - Lh + \gamma +Lh)V = IV + \frac{1}{L}(\gamma+Lh)V \leq IV,
\]
which is due to $I= \frac{\phi}{L}-h$, and $V<0$, $\frac{1}{L}(\gamma+Lh)\geq 0$. 
On the other hand, when $V>0$ (i.e. the upward path along the right line of backlash), the RHS of \eqref{eq:sdot1} becomes
\[
\frac{1}{L}(\phi + \gamma)V = \frac{1}{L}(\phi + Lh + \gamma -Lh)V = IV + \frac{1}{L}(\gamma-Lh)V \leq IV,
\]
where we have used $I = \frac{\phi}{L}+h$,  
$V>0$ and $\frac{1}{L}(\gamma-Lh)\leq 0$. 

\medskip
\noindent\textit{Case 2: $\phi< -\frac{Lh-\gamma}{2}$}, it can be checked that
\begin{equation}\label{eq:sdot2}
\frac{d}{dt}S_{\gamma}(\phi)=\frac{1}{L}(\phi-\gamma)\dot{\phi}=\frac{1}{L}(\phi-\gamma)V.
\end{equation}
Similar as before, $V=0 \Rightarrow \frac{d}{dt}S_\gamma(\phi)=0$. When $V<0$ (along the downward path), the RHS of \eqref{eq:sdot2} satisfies
\[
\frac{1}{L}(\phi - \gamma)V = \frac{1}{L}(\phi -Lh - \gamma+Lh)V = IV +\frac{1}{L}(-\gamma+Lh)V \leq IV,
\]
where $I=\frac{\phi}{L}-h$, $V<0$, and $(-\gamma+Lh)V\leq0$. Finally, when $V>0$ (along the upward path), the RHS of \eqref{eq:sdot2} becomes
\[
\frac{1}{L}(\phi-\gamma)V=\frac{1}{L}(\phi+Lh-\gamma-Lh))V=IV-\frac{1}{L}(\gamma+Lh)V\leq IV,
\]
where $I=\frac{\phi}{L}+h$, $V>0$ and $(\gamma+Lh)V\geq0$. 
\end{proof}

\medskip
One can check that along any closed trajectory in the $(I,\phi)$ plane between $\phi_1$ and $\phi_2 \geq \phi_1$, the total \emph{dissipated energy} per cycle is given by
\[
2h(\phi_2 - \phi_1).
\]
This expression is \emph{independent} of $\gamma \in [-Lh,Lh]$, and equals the area enclosed by the trajectory.
Hence the parameter $\gamma$ only shows up in the expression of admissible storage functions.

The family of storage functions $S_\gamma$ also satisfies a monotonicity (interlacing) property with respect to $\gamma$, namely, for any given $\gamma_1 \leq \gamma_2$, $S_{\gamma_1}(\phi) \leq S_{\gamma_2}(\phi)$ holds for all $\phi$. 
Figure~\ref{fig:admissible_storage_function} shows the plot of $S_{-Lh}$, $S_{Lh}$ and the shaded gray area encapsulating the family of storage functions $S_\gamma$ in~\eqref{eq:S_gamma}. This property is formalized in the following proposition.

\begin{figure}[b]
    \centering
    \includegraphics[]{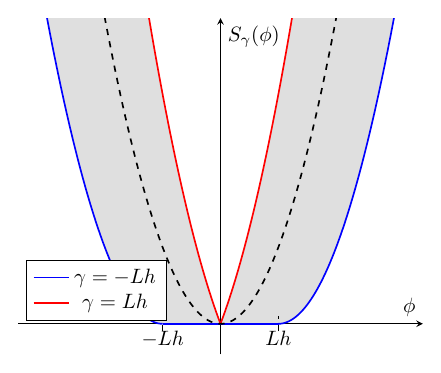}
    \caption{Admissible storage functions $S_\gamma(\phi)$ for the backlash inductor. The extreme cases $\gamma=\pm Lh$ bound the family, while the shaded region represents 
    all admissible functions for $\gamma\in[-Lh,Lh]$. All functions attain a minimum value of zero. The dashed line indicates the admissible Hamiltonian of the backlash inductor, derived in \eqref{eq:Hbf}.}
    \label{fig:admissible_storage_function}
\end{figure}
\begin{prop}\label{lem:monotonicity}
    Let $S_\gamma$ be defined by \eqref{eq:S_gamma}. Then for any $\gamma_1,\ \gamma_2$ with $-Lh\leq\gamma_1<\gamma_2\leq Lh$, 
    \begin{equation}\label{eq:S_ineq} 
    S_{\gamma_1}(\phi)\leq S_{\gamma_2}(\phi) \quad \text{for all $\phi$}.
    \end{equation}
\end{prop}
\begin{proof}
We distinguish four cases based on the value of $\phi$ of our backlash inductor as shown in Figure \ref{fig:backlash}.

\noindent\textit{Case 1:} $\phi > \frac{Lh-\gamma_1}{2}$. Both $S_{\gamma_1}(\phi)$ and $S_{\gamma_2}(\phi)$ are given by the first branch of \eqref{eq:S_gamma}. Let us define
\begin{equation*}
    g(\gamma) := \frac{1}{2L}\left[(\phi+\gamma)^2 - \left(\frac{Lh+\gamma}{2}\right)^2\right].
\end{equation*}
To show $S_{\gamma_1}(\phi) \leq S_{\gamma_2}(\phi)$, we only need to show that $g$ is strictly increasing in $\gamma$ on $[-Lh, Lh]$. Computing the gradient of $g$ (the $\phi^2$ term is independent of $\gamma$ and drops out) yields
\begin{align*}
    \frac{dg}{d\gamma} &= \frac{1}{2L}\left[2\phi + \frac{3}{2}\gamma - \frac{Lh}{2}\right]\\ &> \frac{1}{2L}\left[ Lh - \gamma + \frac{3}{2}\gamma - \frac{Lh}{2}\right] = \frac{Lh+\gamma}{4L} \geq 0,
\end{align*}
where we used $2\phi > Lh - \gamma$ and $\gamma \geq -Lh$. Hence $g$ is strictly increasing in $\gamma$, and $S_{\gamma_1}(\phi) \leq S_{\gamma_2}(\phi)$ follows since $\gamma_1 < \gamma_2$ implies $g(\gamma_1) < g(\gamma_2)$.
\medskip

\noindent\textit{Case 2: $\frac{Lh - \gamma_2}{2} < \phi \leq \frac{Lh - \gamma_1}{2}$.} By the third branch of \eqref{eq:S_gamma}, $S_{\gamma_1}(\phi) = 0$, while by the first branch, $S_{\gamma_2}(\phi) \geq 0$. Hence $S_{\gamma_1}(\phi) \leq S_{\gamma_2}(\phi)$, so \eqref{eq:S_ineq} holds. 

\medskip
\noindent\textit{Case 3: $|\phi| \leq \frac{Lh-\gamma_2}{2}$}. 
In this case, both storage functions are zero. 

\medskip
\noindent\textit{Case 4: $\phi<-\frac{Lh-\gamma_2}{2}$}. For this final case, the proof follows similarly as before due to the 
symmetry in \eqref{eq:S_gamma}. 
\end{proof}

\medskip

\noindent

In the following discussion, we consider the notion of available storage and required supply functions for the backlash inductor, as introduced before in Section \ref{sec:dissipativity}. 
The available storage function corresponds to the maximum energy that can be extracted from the system starting at any state $\phi_0\in\rline$. 
The general definition \eqref{eq:available_storage} is then implemented for the backlash operator as follows.
The \textit{available storage function} $S^a$ of the backlash inductor is defined by 
    \begin{equation}\label{eq:avstorage}
        S^{a}(\phi_0)=\sup_{\substack{V(\cdot),\ T\geq0}}-\int^T_0I(t)V(t)dt
    \end{equation}
with $V(t) = \dot{\phi}$ and $\phi(0)=\phi_0$. 

\begin{prop}
For the backlash inductor with width $2h$ and inductance $L$, the available storage function is given by
$S^a(\phi_0) = S_{-Lh}(\phi_0)$, 
where $S_{-Lh}$ is as in \eqref{eq:S_gamma} with $\gamma=-Lh$, i.e.,  
\begin{equation}\label{eq:availablestorage}
S^a(\phi_0)=\begin{cases}
    \frac{1}{2L}(\phi_0-Lh)^2 & 
    \qquad \forall\ \phi_0> Lh,\\
    \frac{1}{2L}(\phi_0+Lh)^2 & 
    \qquad \forall\ \phi_0<-Lh, \\
    0 & \qquad \text{\rm elsewhere}. 
\end{cases}  
\end{equation}
\end{prop}
\begin{proof}

We will prove this proposition by constructing an input signal $V$ that maximizes the energy extraction without inducing energy dissipation along the path.  
Let us consider $\phi_0>Lh$ and an external port signal $V(t)=\frac{\dd \phi(t)}{\dd t}\leq 0$ such that $\phi(0)=\phi_0$, $\phi(T)=Lh$ and $V(t)=0$ for all $t\geq T$. In this case, the trajectory follows the line $I(t)= \frac{\phi(t)}{L}-h$ for all $t\in [0,T]$ (cf. Figure \ref{fig:backlash}) and it holds that
\begin{align}
\nonumber & -\int_0^TI(t)V(t)\dd t = -\int_0^T\frac{1}{L}(\phi(t)-Lh)V(t)\dd t \\
 & = -\frac{1}{L}\int_{\phi_0}^{Lh}(\phi-Lh)\dd\phi 
\label{eq:integral_Sa} = \frac{1}{2L}(\phi_0-Lh)^2. 
\end{align}
Let us now consider port signal $V$ that first increases $\phi$ to $\phi_0+\epsilon$ by taking $V(t)>0$ for $t\in [0,t_1]$, then applying $V(t)\leq 0$ as before for $t\in [t_1,T]$ such that $\phi(T)=Lh$, and thereafter $V(t)=0$ for all $t>T$. For this case, the trajectory of $I$ satisfies $I(t)= \frac{\phi(t)}{L}+h$ for all $t\in [0,t_1]$ and $I(t)= \frac{\phi(t)}{L}-h$ for all $t\in [t_1,T]$ (cf. Figure \ref{fig:backlash}). Accordingly, it can be calculated that
\begin{align}
\nonumber &-\int_0^TI(t)V(t)\dd t  
= -\int_0^{t_1}\frac{1}{L}(\phi(t)+Lh)V(t)\dd t \\
\nonumber &\qquad \qquad \qquad \qquad - \int_{t_1}^{T}\frac{1}{L}(\phi(t)-Lh)V(t)\dd t \\
\nonumber &= -\frac{1}{L}\int_{\phi_0}^{\phi_0+\epsilon}(\phi+Lh)\dd\phi 
           -\frac{1}{L}\int_{\phi_0+\epsilon}^{Lh}(\phi-Lh)\dd\phi
\end{align}
\begin{align}
\nonumber &= -\frac{1}{L}\int_{\phi_0}^{\phi_0+\epsilon}(\phi+Lh)\dd\phi 
           -\frac{1}{L}\int_{\phi_0+\epsilon}^{\phi_0}(\phi-Lh)\dd\phi \\
\nonumber &\quad -\frac{1}{L}\int_{\phi_0}^{Lh}(\phi-Lh)\dd\phi \\
\label{eq:dissipation_term_in_av}
          &= \frac{1}{2L}(\phi_0-Lh)^2 - \Delta(\epsilon),
\end{align}
where $\Delta(\epsilon) = \frac{1}{L}\int_{\phi_0}^{\phi_0+\epsilon}(\phi+Lh)\dd\phi - \frac{1}{L}\int_{\phi_0}^{\phi_0+\epsilon}(\phi-Lh)\dd\phi = 2\epsilon h >0$.  
Correspondingly, since $\Delta(\epsilon)>0$, it is clear from the two cases of port signal $V(t)$ as above that evaluating \eqref{eq:avstorage} corresponds to the first choice of $V(t)$ that gives us \eqref{eq:integral_Sa}, i.e., the first case in \eqref{eq:availablestorage} holds. Furthermore, any trajectory where $\phi(T) < Lh$ stays on the line 
$I(t) = \frac{\phi(t)}{L}-h$ past $\phi = Lh$, where $I(t) < 0$ 
and $V(t) \leq 0$, so $I(t)V(t) \geq 0$; hence extending the trajectory 
only reduces the extracted energy. Together with the argument above, this confirms that \eqref{eq:integral_Sa} is indeed the supremum. 
The same argumentation applies also when $\phi_0<-Lh$, in which case the second case in \eqref{eq:availablestorage} is also satisfied.

It remains to prove the last case in \eqref{eq:availablestorage}, where 
$\phi_0 \in [-Lh,Lh]$. 
For any $V = \dot{\phi} \neq 0$, the backlash characteristic in Figure \ref{fig:backlash} shows that $I(t)$ immediately takes the same sign as $V(t)$, so that $I(t)V(t) \geq 0$ for all $t \geq 0$. Therefore
\[
-\int_0^T I(t)V(t)\,\dd t \leq 0,
\]
with equality only when $V = 0$. Since the supremum in \eqref{eq:avstorage} is taken over all admissible inputs, the maximum is attained by $V = 0$ so that $S^a(\phi_0) = 0$. 

\end{proof}

As discussed before, the available storage function corresponds to the maximum energy that can be extracted from the backlash operator starting at any given state $\phi_0$. The dual notion, the required supply function \eqref{eq:required_supply} corresponds to the minimum energy needed to bring the system from the ground state $\phi(-T)=\phi^*$ to a given state $\phi(0) = \phi_0$, where we can take any point $\phi^*$ from the set $\mathcal X^*:=\{(I,\phi)\, |\, \phi = 0, I\in [-h,h]\}$ as the ground state. In the following definition and proposition, we use the origin as the ground state $x^*$. 
The \textit{required supply} $S^r$ of the backlash inductor is defined by
\begin{equation}\label{eq:required_supply_def}
    S^{r}(\phi_0) = \inf_{\substack{V(\cdot),\, T\geq 0}} 
    \int_{-T}^{0} I(t)V(t)\,\dd t,
\end{equation}
with $V(t) = \frac{\dd\phi(t)}{\dd t}$ and  $\phi(0)=\phi_0$, with the ground state $\phi(-T)= 0$. 


\begin{prop}
For the backlash inductor, the required supply $S^r$ as in \eqref{eq:required_supply_def} 
is given by 
\[
S^r(\phi_0) = S_{Lh}(\phi_0),
\] 
where $S_{Lh}$ is as in \eqref{eq:S_gamma} with $\gamma=Lh$, i.e.,
\begin{equation}\label{eq:requiredsupply}
    S^r(\phi_0)=\begin{cases}
        \frac{1}{2L}(\phi_0+Lh)^2 - \frac{1}{2}Lh^2 & \text{if} \quad \phi_0\geq 0,\\
        \frac{1}{2L}(\phi_0-Lh)^2 - \frac{1}{2}Lh^2 & \text{if} \quad \phi_0 < 0.
    \end{cases}
\end{equation}
\end{prop}
\begin{proof}
We construct admissible port signals $(I,V)$ that achieve the infimum in 
\eqref{eq:required_supply_def} for any given state $\phi_0$, starting 
from the given ground state $\phi^*=0$. 

For the case $\phi_0\geq 0$, the infimum in \eqref{eq:required_supply_def} is achieved by the monotonically increasing trajectory with $V(t)>0$ for all $t\in[-T,0]$, along which $I(t)=\frac{\phi(t)}{L}+h$. Indeed, any trajectory that is not monotone introduces additional dissipation $\Delta(\epsilon)>0$ as shown in \eqref{eq:dissipation_term_in_av}, and is therefore not optimal. The required supply along the optimal path is
\begin{align}
\nonumber \int_{-T}^0 I(t)V(t)\,\dd t &= \int_{\phi^*}^{\phi_0} 
\frac{\phi+Lh}{L}\,\dd\phi \\
&= \frac{1}{2L}(\phi_0+Lh)^2 - \frac{1}{2L}(\phi^*+Lh)^2,
\end{align}
which, after substituting $\phi^*=0$, 
yields the first case of \eqref{eq:requiredsupply}. 

The case $\phi_0 < 0$ follows by symmetry; an analogous argument with $V(t)<0$ and $I(t)=\frac{\phi(t)-Lh}{L}$ yields the second case of \eqref{eq:requiredsupply}. 
\end{proof} 

Following the available storage and required supply bounds as in \eqref{eq:sandwich}, any admissible storage 
function $S$ satisfies
\begin{equation}\label{eq:Sa_Sr_sandwich}
S^a(\phi) \leq S(\phi) \leq S^r(\phi), 
\end{equation}
where $S^a = S_{-Lh}$ and $S^r = S_{Lh}$ are the extremal elements of the family 
$S_\gamma$ defined in \eqref{eq:S_gamma}. Here we have used the fact that $S^a(\phi^*)=0$. 

\medskip
\section{pH formulation}
In this section, we derive a pH representation of the backlash inductor. Furthermore, it turns out that this is given as the parallel interconnection of a linear inductor and a multi-valued resistor. Using the modular property of pH systems, namely that the power-preserving interconnection of pH subsystems is again a pH system \cite{van2014port}, we then present its parallel and series interconnection with a linear resistor and capacitor.

\subsection{Backlash inductor in pH form}

Consider the backlash diagram in the $(I,\phi)$ plane as shown in 
Fig.~\ref{fig:backlash}. 
Using the multivalued sign mapping as in \eqref{eq:sign_V}, a pH representation of the backlash inductor can be given by
\begin{equation}\label{eq:pHbacklash}
\begin{aligned}
\dot{\phi} &= V,\\
I &\in \frac{\partial H_{bf}(\phi)}{\partial \phi}
     + h\sign(V),
\end{aligned}
\end{equation}
with Hamiltonian

\begin{equation}\label{eq:Hbf}
H_{bf}(\phi) = \frac{1}{2L}\phi^2,
\end{equation}

i.e., the magnetic energy of the linear inductor component. Clearly $S^a(\phi) \leq H_{bf}(\phi) \leq S^r(\phi)$, as illustrated by the dashed line in Fig. \ref{fig:admissible_storage_function}. It is immediately seen that \eqref{eq:pHbacklash} reproduces the backlash characteristic of Fig. \ref{fig:backlash}. 


The first term in \eqref{eq:pHbacklash} corresponds to the energy-storing linear inductance, while the second one is the term that induces 
hysteretic dissipation. The latter acts as a nonlinear feedthrough mapping and is structurally analogous to Coulomb friction. 

The energy balance follows directly from \eqref{eq:pHbacklash}, i.e.,
\begin{equation}\label{eq:bf_energy}
\frac{d}{dt}H_{bf}(\phi)
= \frac{\partial H_{bf}(\phi)}{\partial \phi}\dot{\phi}
= \Big(I - h\sign(V)\Big)V,
\end{equation}
and since $h\sign(V)V \geq 0$, it follows that 
\[
\frac{d}{dt}H_{bf}(\phi) \le IV.
\]

This shows that \eqref{eq:pHbacklash} is passive with respect to the port variables $(V,I)$. 

The model in \eqref{eq:pHbacklash} can be rewritten as follows, using convex analysis. Note that the monotone multi-valued sign mapping can be expressed as the sub-differential of the convex function 
\begin{equation} \label{eq:nonlineardissipation} 
P(V):=h|V|,
\end{equation}
given by $\frac{\partial P(V)}{\partial V} =h\sign(V)$. Consequently, the pH system of \eqref{eq:pHbacklash} can be also written as 

\begin{equation}\label{eq:pHfeedthrough}
\begin{aligned}
\dot{\phi} &= V,\\
I &\in \frac{\partial H_{bf}(\phi)}{\partial \phi}
     + \frac{\partial P(V)}{\partial V} .
\end{aligned}
\end{equation}
This corresponds to the pH system formulation \eqref{eq:pHsystem} with nonlinear dissipation, where $x=\phi$, $u=V$, $y=I$, $J=0$, $G=1$, and $M=0$.
\subsection{RC network with a backlash inductor}
Let us now consider a network consisting of a backlash inductor interconnected to a linear resistor of resistance $R$ and a linear capacitor of capacitance $C$. We derive the pH formulation for both the series and parallel interconnection configurations as shown in Fig.~\ref{fig:RLCbacklashparralel} and Fig.~\ref{fig:RLCbacklashseries}, respectively. The total Hamiltonian of the network is given as
\begin{equation}\label{eq:hamiltonian}
H(\phi, Q) := H_{bf}(\phi) + H_C(Q),
\end{equation}
where $H_{bf}(\phi)$ is as in \eqref{eq:Hbf} and $H_C(Q)=\frac{1}{2C}Q^2$ is the energy stored in the capacitor with charge $Q$. This Hamiltonian holds for both the parallel and series configurations of the network. 

\begin{figure}[b]
    \centering
    \begin{circuitikz}[american]
        \draw
            (0,0) to[isource, l=$I$] (0,3)
            (0,3) -- (6,3)
            (0,0) -- (6,0)
            (2,3) to[R, l=$R$] (2,0)
            (4,3) to[L, l=$L$] (4,0)
            (6,3) to[C, l=$C$] (6,0);
    \end{circuitikz}
    \caption{Parallel interconnected resistor,  backlash inductor and capacitor circuit.}
    \label{fig:RLCbacklashparralel}
\end{figure}
\subsubsection{Parallel interconnection} 
In the parallel configuration, as shown in Fig.~\ref{fig:RLCbacklashparralel}, all elements share the same voltage $V$. The source is a current source with input $u = I$ and output $y = V$, so that the supply rate is $I V$. It can be computed that the corresponding pH formulation is given by

\begin{equation}
\begin{split}
\begin{bmatrix}
    \dot{\phi}\\
    \dot{Q}
\end{bmatrix}&\in\begin{bmatrix}
    0 & 1\\
    -1 & 0
\end{bmatrix}\begin{bmatrix}
    \frac{\partial H(\phi,Q)}{\partial \phi}\\
    \frac{\partial H(\phi,Q)}{\partial Q}
\end{bmatrix}-\begin{bmatrix}
    0\\
    \frac{\partial \tilde{P}(V)}{\partial V} 
\end{bmatrix}+\begin{bmatrix}
    0\\
    1
\end{bmatrix} I,\\
V&=\begin{bmatrix}
    0 & 1
\end{bmatrix}
\begin{bmatrix}
\frac{\partial H(\phi,Q)}{\partial \phi}\\
\frac{\partial H(\phi,Q)}{\partial Q}
\end{bmatrix},
\end{split}
\end{equation}
where $\tilde{P}(V) := P(V) + \frac{V^2}{2R}$ is the dissipation potential of the network, with $P$ as in \eqref{eq:nonlineardissipation}. This system corresponds to the pH system \eqref{eq:pHsystem}, with state $x=\begin{bmatrix}
    \phi & Q
\end{bmatrix}^\top$ with 
\[
J = \begin{bmatrix}0&1\\-1&0\end{bmatrix}, \quad
G = \begin{bmatrix}0\\1\end{bmatrix}, \quad
M = 0.
\]

The energy balance follows immediately as
\begin{equation}
\dot{H}(\phi,Q) = -\frac{V^2}{R} - V\cdot h\sign(V) + V I \leq V I,
\end{equation}
which shows the passivity of the network with respect to the port $(I, V)$. The two dissipation terms $\frac{V^2}{R}\geq0$ and $V\cdot h\sign(V)\geq0$ represent resistive and hysteretic losses, respectively, and together form $V\frac{\partial \tilde{P}(V)}{\partial V}\geq0$. Note that the dissipative current $I_d = \frac{\partial \tilde{P}(V)}{\partial V} = h\sign(V) + \frac{V}{R}$ results in the following dead-zone characteristic
\[
V = \begin{cases}
R(I_d-h) & \forall\, I_d > h,\\
0        & \forall\, |I_d| \le h,\\
R(I_d+h) & \forall\, I_d < -h,
\end{cases}
\]
which is the inverse of the `Coulomb friction' basic hysteretic element of \cite[Fig.~3]{karnopp1983computer}, with slope $1/R$.
\subsubsection{Series interconnection}

\begin{figure}[h]
    \centering
        \begin{circuitikz}[american]
            \draw
                (0,0) to[vsource, l=$V$, invert] (0,3)
                to[R, l=$R$] (3,3)
                to[L, l=$L$] (6,3)
                to[C, l=$C$] (6,0)
                -- (0,0);
\end{circuitikz}
\caption{Series interconnected resistor, backlash inductor and capacitor circuit.}
    \label{fig:RLCbacklashseries}
\end{figure} 
We now consider the series configuration as shown in
Fig.~\ref{fig:RLCbacklashseries}, where all elements share the same current
$I$, with $V = V_R + V_L + V_C$ by Kirchhoff's voltage law, where
$V_L = \dot{\phi}$, $V_R = RI$ and $V_C = \frac{Q}{C}$. The source is a voltage
source with input $u = V$ and output $y = I$, so that the supply rate is $IV$.
In this case, the corresponding pH formulation is given by

\begin{equation}\label{eq:pHseries}
\begin{aligned}
\begin{bmatrix}
    \dot{\phi}+R\dot{Q}\\
    \dot{Q}-h\sign(V_L)
\end{bmatrix} &\ni    
\begin{bmatrix}
0 & -1\\
1 & 0
\end{bmatrix}\begin{bmatrix}
    \frac{\partial H(\phi,Q)}{\partial \phi}\\
    \frac{\partial H(\phi,Q)}{\partial Q}
\end{bmatrix}+\begin{bmatrix}
    1\\
    0
\end{bmatrix}V,\\
I&\in\begin{bmatrix}
    1 & 0
\end{bmatrix}\begin{bmatrix}
    \frac{\partial H(\phi,Q)}{\partial \phi}\\
    \frac{\partial H(\phi,Q)}{\partial Q}
\end{bmatrix}+\frac{\partial P(V_L)}{\partial V_L},
\end{aligned}
\end{equation}
where $\dot{Q}=I$, so that the first row expresses $V_L + V_R = V - V_C$ and the second row the backlash characteristic $\frac{\partial H_{bf}}{\partial\phi} =I - h\sign(V_L)$. The resulting system is an implicit pH system, as both the backlash term and the resistive drop act through the state derivative. 
Nevertheless, the interconnection matrix is skew-symmetric, and the standard pH power balance is preserved, as follows from computing the time derivative of \eqref{eq:hamiltonian},
\begin{equation}\label{eq:energybalanceseries}
\begin{aligned}
\dot{H}(\phi,Q) &= \Big(I - h\sign(V_L)\Big)V_L + V_C I\\
                &= I\,(V_L + V_C) - h|V_L|\\
                &= IV - RI^2 - h|V_L| \leq IV.
        \end{aligned}
\end{equation}
The two dissipation terms $RI^2 \geq 0$ and $h|V_L| \geq 0$ represent the
resistive and hysteretic losses, respectively. Hence both configurations are
passive with respect to the port variables $(V,I)$.
\section{Conclusion}
In this paper, we have shown that the backlash inductor admits a family of storage functions that are lower-bounded by the available storage and upper-bounded by the required supply, as illustrated in Fig.~\ref{fig:admissible_storage_function}. The  backlash inductor can then be expressed in a pH framework with a nonlinear dissipation feedthrough term, allowing the system to be written in standard nonlinear pH form, while preserving the skew-symmetric interconnection structure and energy balance. The formulation is subsequently extended to both parallel and series RC network configurations incorporating the backlash inductor. For the resulting network, passivity with respect to the supply rate is established using the total Hamiltonian. Future work will address the generalization to other hysteretic storage elements, where the hysteresis is described by the Prandtl--Ishlinskii, Preisach or Duhem models.
\bibliographystyle{IEEEtran}

\bibliography{References}
\end{document}